\documentclass{llncs}
\usepackage[utf8]{inputenc}

\usepackage{amsmath}
\usepackage{amssymb}
\usepackage{amsfonts}
\usepackage{amsthm}
\usepackage{graphicx}
\usepackage{epstopdf}
\usepackage{enumerate}
\usepackage[ruled,vlined,linesnumbered]{algorithm2e}
\usepackage[colorinlistoftodos]{todonotes}
\usepackage{soul}

\usepackage[a4paper,top=3cm,bottom=3cm,left=3cm,right=3cm,marginparwidth=4cm]{geometry}

\newcommand{\Fail}{\textsc{Fail}}
\newcommand{\None}{\textsc{None}}

\definecolor{lightgreen}{RGB}{180,255,150}
\definecolor{lightblue}{RGB}{180,180,255}

\title{Combining Networks using Cherry Picking Sequences\thanks{Research funded by the Netherlands Organization for Scientific Research (NWO), with the Vidi grant 639.072.602.}}
\author{Remie Janssen\orcidID{0000-0002-5192-1470}\inst{1} \and Mark Jones\orcidID{0000-0002-4091-7089}\inst{2} \and Yukihiro Murakami\orcidID{0000-0003-1355-5884}\inst{1}}
\authorrunning{R. Janssen et al.}
\institute{Delft Institute of Applied Mathematics, Delft University of Technology, Van Mourik Broekmanweg 6, 2628 XE Delft, The Netherlands
\email{\{R.Janssen-2, Y.Murakami\}@tudelft.nl} \and Centrum Wiskunde \& Informatica,  P.O.Box 94079, 1090 GB Amsterdam, The Netherlands \email{markelliotlloyd@gmail.com}}

\begin{document}
\maketitle
\begin{abstract}
Phylogenetic networks are important for the study of evolution. The number of methods to find such networks is increasing, but most such methods can only reconstruct small networks. To find bigger networks, one can attempt to combine small networks. In this paper, we study the {\sc Network Hybridization} problem, a problem of combining networks into another network with low complexity. We characterize this complexity via a restricted problem, {\sc Tree-child Network Hybridization}, and we present an FPT algorithm to efficiently solve this restricted problem.
\end{abstract}

\section{Introduction}
Evolutionary histories are often represented by phylogenetic trees, and more recently, by phylogenetic networks. 
Knowing the evolutionary history of a species is vital for understanding their biology. 
Therefore, it is important to have methods for finding phylogenetic networks that accurately represent the true evolutionary histories. Many methods exist to find evolutionary histories; some are purely combinatorial, others have a likelihood component as well. Here, we focus mainly on the purely combinatorial problems. 

One classic combinatorial method is to solve {\sc Hybridization}: given a set of trees, find the simplest network that displays these trees \cite{baroni2005framework}. 
Unfortunately, the problem is NP-hard, even on inputs of two trees \cite{bordewich2007computing}.
For this problem, it is assumed we can construct accurate phylogenetic trees for small parts of the genomes of the studied taxa. When the input consists of only two or three trees, it can be solved relatively efficiently---EPT time \cite{whidden2013fixed,van2016hybridization}---even though the problem is already NP-hard in that case. For an input consisting of three trees or more, there is still an FPT algorithm \cite{van2013quadratic}, but it is not practical. In these cases, it is useful to limit the search space to networks with a restricted structure, such as tree-child networks \cite{van2019practical}, or temporal networks \cite{humphries2013cherry}.

Another combinatorial approach for finding phylogenetic networks is to combine smaller networks. 
The smaller networks are often assumed to have certain properties. For example, it may be assumed that we are given all strict subnetworks containing the full set of leaves. In that case, it is possible to reconstruct a level-$k$ tree-child network from all its level-$(k-1)$ subnetworks \cite{murakami2019reconstructing}. Another assumption could be that the input consists of all subnetworks obtained by removing exactly one leaf \cite{van2018leaf}. Instead of having almost all leaves, the subnetworks can also be allowed to have only few leaves. For example, low level networks can be reconstructed from their full set of binets \cite{huber2017reconstructing,van2017binet}, trinets \cite{oldman2016trilonet,van2014trinets} or quarnets \cite{huber2018quarnet}.

In practice, it may not be easy to find \textit{all} subnetworks. This renders many of the previously mentioned methods useless. Furthermore, these methods typically only work for low level networks. This means that they cannot be used when the phylogenetic signal comes from a complicated evolutionary history, or if there is some randomness in the data, complicating the data as well.

In this paper, we combine networks that all contain the full set of leaves, but we do not assume we have all the subnetworks of the network we want to find. The problem we solve is analogous to {\sc Hybridization}, but with networks as the input, {\sc Network Hybridization}: Given a set of networks with taxa $X$, find a network $N$ with minimal reticulation number, that displays all input networks. 
Since this is a generalization of the {\sc Hybridization} problem, the problem remains NP-hard in general, even for inputs of two networks.
We show that for the restricted problem on tree-child (topologically restricted class of networks; defined formally in Section~\ref{sec:Preliminaries}) inputs and output, we can use tree-child sequences to obtain an FPT algorithm.
This FPT algorithm is an extension of the one introduced in \cite{van2019practical} in which they considered tree inputs; the tree-child sequence approach was first introduced in \cite{linz2019attaching}.
We also comment briefly on how some measure of an optimal solution to the {\sc Network Hybridization} problem can be characterized by solving this restricted problem.

\paragraph{Structure of the paper}
We start with a quick introduction of relevant concepts from mathematical phylogenetics in Section~\ref{sec:Preliminaries}. 
Then, in Section~\ref{sec:NwHyb}, we formally introduce {\sc Tree-child Network Hybridization} and prove its relation to tree-child sequences.
This section also relates this problem to the more general {\sc Network Hybridization}.
In Section~\ref{sec:CountingCherries}, we lay the theoretical foundation to extend the algorithm in \cite{van2019practical} that takes inputs of trees to also work for inputs of networks.
As a last theoretical section in the paper, we present an FPT algorithm that solves {\sc Tree-Child Network Hybridization} (Section~\ref{sec:Algorithms}).
We conclude the paper with a discussion, including some open questions (Section~\ref{sec:Discussion}).

\section{Preliminaries}\label{sec:Preliminaries}
The main objects of study for this paper are phylogenetic networks. These graphs are used in biology to represent evolutionary scenarios for a given set of species.

\begin{definition}
A \emph{(rooted phylogenetic) network} on a finite set of \emph{taxa} $X$ is a directed acyclic graph with 
\begin{itemize}
    \item one node of indegree-0 and outdegree-1, the \emph{root};
    \item nodes of indegree-1 and outdegree-0 labelled bijectively with $X$, the \emph{leaves};
    \item nodes of indegree-1 and outdegree-2, the \emph{tree nodes};
    \item nodes of indegree greater than 1 and outdegree-1, the \emph{reticulations}.
\end{itemize}
\end{definition}

If all the reticulation nodes have indegree-2, the network is called \emph{binary}.
An edge $uv$ is called a \emph{tree edge} if $v$ is a tree node or leaf, and a \emph{reticulation edge} if $v$ is a reticulation.
The vertex~$u$ is the \emph{parent} of~$v$, and~$v$ is the \emph{child} of~$u$.
The \emph{reticulation number}~$r(N)$ of a network~$N$ is the total number of reticulation edges minus the total number of reticulations.

A network is \emph{stack-free} if every reticulation has a child that is a tree node or a leaf.
A network is \emph{tree-child} if it is stack-free and every tree node has a child that is a tree node or a leaf.
We now define some relevant notation for local structures in phylogenetic networks.

\begin{definition}
Let $N$ be a network on $X$ and $x,y\in X$ two leaves. Then we say $N$ has a \emph{cherry} $\{x,y\}$ if the parent of $x$ is the parent of $y$; we say that $N$ has a \emph{reticulated cherry} $(x,y)$ if the parent of $x$ is a reticulation, and the parent of $y$ is a parent of this reticulation.
If~$(x,y)$ is a cherry or a reticulated cherry in~$N$, then it is called a \emph{reducible pair}.
\end{definition}

Tree-child sequences are built on the notion of \emph{reducing} cherries and reticulated cherries from networks.

\begin{definition}
Let $N$ be a network on $X$, and $(x,y)$ a pair of leaves.
Let~$p_x$ and~$p_y$ denote the parents of~$x$ and~$y$ in~$N$, respectively
Then \emph{reducing} $(x,y)$ in $N$ results in a network $N(x,y)$ obtained as follows:
\begin{itemize}
    \item If $\{x,y\}$ is a cherry in $N$, remove $x$ and the pendant edge $p_xx$, and suppress $p_x$ if it has become a degree-2 node;
    \item If $(x,y)$ is a reticulated cherry in $N$, remove the reticulation edge $p_yp_x$ and suppress $p_x$ or $p_y$ if it has become a degree-2 node.
    \item Otherwise, $N(x,y):=N$.
\end{itemize}
\end{definition}

The reversal of reducing cherries and reticulated cherries can be done by \emph{adding} ordered pairs of leaves to the network.

\begin{definition}\label{def:Const}
Let~$N$ be a network and let~$(x,y)$ be reducible pair.
Then we may \emph{construct}~$N$ from~$N(x,y)$---also called \emph{add~$(x,y)$ to~$N(x,y)$}---by applying the following.
\begin{enumerate}
    \item If $x$ is a leaf in~$N(x,y)$ (i.e., if~$(x,y)$ is a reticulated cherry in~$N$), and
    \begin{enumerate}
        \item if $p$, the parent of $x$ in $N(x,y)$, is a reticulation then add a node $q$ directly above $y$, and add an edge $qp$.\label{def:ConstNoStack}
        \item otherwise, add nodes~$p$ and $q$ directly above $x$ and $y$ respectively, and add an edge $qp$.
    \end{enumerate}
    \item If $x$ is not a leaf in $N(x,y)$ (i.e., if~$(x,y)$ is a cherry in~$N$) then add a labelled node $x$, insert a node $q$ directly above $y$, and add an edge $qx$.
\end{enumerate}
\end{definition}

The above notion of adding an ordered pair of leaves~$(x,y)$ to a network~$N$ is well-defined if~$y$ is already a leaf in~$N$.
If this is indeed the case, we may obtain a network from a sequence of ordered pairs by iteratively adding ordered pairs to an existing network.
To do so, we impose the following condition on our sequence of ordered pairs:
\emph{The second coordinate of each pair has to occur as a first coordinate in the remainder of the sequence, or as the second coordinate of the last pair.}
Then, the following procedure constructs a network from a sequence.

\begin{procedure}
  \caption{ConstructNetworkFromSequence($S$)\label{alg:construct-network}}
  \KwIn{A sequence of ordered pairs $S=(x_1,y_1)\cdots(x_n,y_n)$\;}
  \KwOut{The network that can be constructed from $S$\;}
  Set $N$ to be the tree on one leaf $y_n$\;
  \For{$i=n,\ldots,1$}{
    \uIf{$x_i$ is a leaf of $N$}
    {
      \uIf{the parent of~$x_i$ is a reticulation}
      {
      Let~$p_x$ denote the parent of~$x_i$\;
      }
      \Else
      {
      Subdivide the incoming edge of $x_i$ with a node $p_x$\;
      }
      Subdivide the incoming edge of $y_i$ with a node $p_y$\;
      Add the edge $p_yp_x$ to $N$\;
    }
    \Else
    {
      Subdivide the incoming edge of $y_i$ with a node $p_y$\;
      Add a new node labelled $x_i$ to $N$\;
      Add the edge $p_yx_i$ to $N$\;
    }
  } 
  \Return $N$\;
\end{procedure}

Note that because we only add reticulation edges to existing reticulation nodes wherever possible (Line~$4$), the network obtained by using the above procedure is always stack-free.
Imposing another condition: \emph{no first coordinate leaf is used as a second coordinate in the remainder of the sequence} on the sequence ensures that the network we obtain is tree-child.
With this in mind, we formally define a tree-child sequence.

\begin{definition}\label{def:CPS}
A \emph{tree-child sequence (TCS)} is a sequence of ordered pairs of two leaves such that the following conditions hold:
\begin{itemize}
    \item the second coordinate of each pair has to occur as a first coordinate in the remainder of the sequence or as the second coordinate of the last pair;
    \item no first coordinate leaf is used as a second coordinate in the rest of the sequence.
\end{itemize}
\end{definition}

Let $S$ be a TCS, that involves the leaves $X$. Then, the \emph{weight} of $S$ is~$w(S) = |S|-|X|+1$.
Given a sequence of ordered pairs~$S=S_1S_2\cdots S_{|S|}$, we let~$NS$ denote the network
\[NS:= (\cdots((NS_1)S_2)\cdots S_{|S|-1})S_{|S|}= NS_1S_2 \cdots S_{|S|}.\]
We introduce some notation for subsequences of a sequence $S$.
For~$i\in[|S|]$, we use the following notation for subsequence of~$S$.
The~$i$th ordered pair of~$S$ is $S_i = (x_i,y_i)$.
The first~$i$ ordered pairs in~$S$ is denoted by~$S_{[:i]} = (x_1,y_1),\dots, (x_i, y_i)$.
The subsequence of~$S$ without the first~$i$ ordered pairs is denoted by~$S_{[i+1:]} = (x_{i+1},y_{i+1}), (x_{i+2},y_{i+2}), \dots, (x_n,y_n)$.
We say that the leaves~$x_1,\ldots, x_i$ are \emph{forbidden} for~$S_{[:i]}$.
Forbidden leaves do not appear as a second coordinate leaf in a TCS (by the second condition of TCSs).

We say $S$ \emph{reduces $N$ to the leaf $x$} if $NS$ is the tree with the single leaf $x$. Similarly, let $\mathcal{N}$ be a set of networks, then we denote by $\mathcal{N}S$ the set of reduced networks $\{NS:N\in\mathcal{N}\}$, and we say $S$ \emph{reduces $\mathcal{N}$ to $x$} if $NS$ is the one leaf tree $x$ for all $N\in\mathcal{N}$.

We call a sequence~$S'$ of ordered pairs a \emph{partial TCS} if there exists a TCS~$S$ such that~$S_{[:i]} = S'$ for some~$i$.

 \begin{figure}[h]
    \centering
    \includegraphics[width=\textwidth]{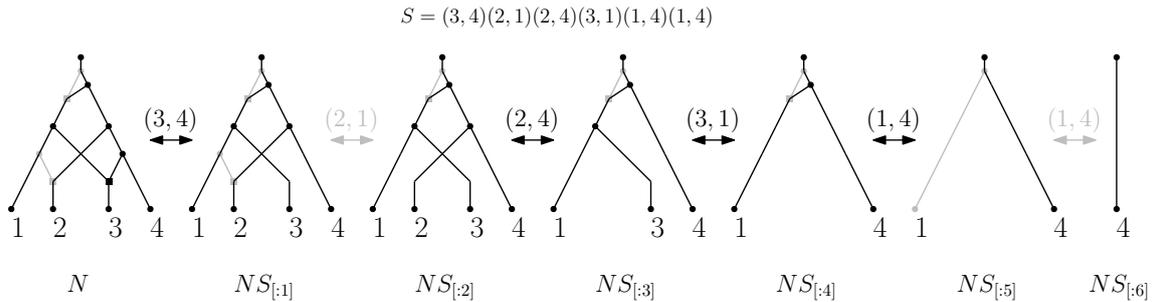}
    \caption{A binary tree-child network~$N$ (grey and black) reduced to a leaf~$4$ by a tree-child sequence~$S$. 
    The reduction is shown as a sequence of networks~$NS_{[:i]}$ for~$i=0,1,\ldots,6$ from left to right, in which an element of~$S$ is applied to the network successively. Each element of $S$ reduces a pair in $N$.
    An example of a cherry~$(3,1)$ can be seen in the network~$NS_{[:3]}$, and a reticulated cherry~$(3,4)$ can be seen in the network~$N$.
    The subnetwork of $N$ consisting of the black edges is also reduced by $S$, and the embedding can be constructed by building both networks simultaneously and keeping track of the edges added by the pairs that change the subnetwork (black pairs and arrows).} 
    \label{fig:SubsequenceEmbedding}
\end{figure}

\section{{\sc Network Hybridization}}\label{sec:NwHyb}
In this section we formally define the {\sc Tree-child Network Hybridization} problem, which asks to find a tree-child network with minimal reticulation number that \emph{displays} all input tree-child networks on the same set of taxa.
We generalize the results presented in~\cite{linz2019attaching} (they considered inputs of trees while we consider inputs of networks) by showing how this problem relates to the more generalized problem of {\sc Network Hybridization} and also to the {\sc Tree-child Weight} problem.
For the {\sc Tree-child Network Hybridization} problem, it turns out that there is not always a solution for some given inputs; we also comment on when this is the case. 

We start by defining what it means for a network to \emph{display} another network.

\begin{definition}
Let~$N$ be a network on the set of taxa,~$X$.
A network~$N'$ on~$Y\subseteq X$ is a \emph{subnetwork} of~$N$ if 
$N'$ can be obtained from~$N$ by deleting reticulation edges, removing leaves not labelled by $Y$, and suppressing all degree-$2$ nodes in the resulting subgraph.
If~$N'$ can be obtained from a subnetwork of~$N$ by contracting edges, then we say~$N$ \emph{displays}~$N'$.
Given a set of networks~$\cal{N}$ on some subsets of the taxa~$X$, then we say that~$N$ \emph{displays}~$\cal{N}$ if~$N$ displays all networks in~$\cal{N}$.
\end{definition}

If a network~$N'$ on~$X$ is a subnetwork of another network~$N$ on~$X$, then an \emph{embedding} of~$N'$ into~$N$ is the mapping of the nodes of~$N'$ to the nodes of~$N$ such that the leaves of~$N'$ are mapped to the leaves of~$N$ with the same labels, and the edges of~$N'$ are mapped to edge-disjoint paths of~$N$.
Our main focus of this paper is to solve the following problem.

\medskip
\fbox{
\parbox{0.8\linewidth}{
{\sc Tree-child Network Hybridization}\\
{\bf Input:} A set of rooted tree-child networks $\mathcal{N}$ on $X$.\\
{\bf Output:} A tree-child network that displays $\mathcal{N}$ with minimal reticulation number if it exists, NO otherwise.}
}
\medskip

Given an optimal tree-child network to the {\sc Tree-child Network Hybridization} problem, one may find a TCS that reduces it.
We will show that the weight of such a TCS is equal to the weight of an optimal solution to the following related problem.

\medskip
\fbox{
\parbox{0.8\linewidth}{
{\sc Tree-child Weight}\\
{\bf Input:} A set of rooted networks $\mathcal{N}$ on $X$.\\
{\bf Output:} A minimal weight TCS that reduces $\mathcal{N}$ if it exists, NO otherwise.}
}
\medskip

Let $\mathcal{N}$ be a set of networks on $X$.
The reticulation number of an optimal solution to {\sc Tree-child Hybridization} is denoted $h_{\mathrm{tc}}(\mathcal{N})$.
The weight of an optimal solution to {\sc Tree-child Weight} is denoted $s_{\mathrm{tc}}(\mathcal{N})$.

For a set of trees~$\mathcal{T}$, the relation $h_{\mathrm{tc}}(\mathcal{T})=s_{\mathrm{tc}}(\mathcal{T})$ holds.
We will extend this result for network inputs. 
We first recall some key lemmas from \cite{YukiRemie}. The first lemma loosely states that 
each TCS reducing a set of networks $\mathcal{N}$ gives a tree-child network with corresponding reticulation number that displays $\mathcal{N}$.
The second lemma gives the reverse statement: 
each tree-child network that displays a set of networks $\mathcal{N}$ gives a TCS of corresponding weight that reduces $\mathcal{N}$.

\begin{lemma}[\cite{YukiRemie},~Lemma~\textbf{8}]\label{lem:reductionImpliesContainment}
Let $N$ and~$N'$ be a tree-child network.
Suppose there is a TCS $S$ that reduces both~$N$ and~$N'$,
such that each element of~$S$ that reduces a pair in~$N'$ also reduces a reducible pair in~$N$.
Then~$N'$ is a displayed by~$N$.
\end{lemma}

\begin{lemma}[\cite{YukiRemie},~Corollary~\textbf{4}]\label{lem:ContainmentImpliesReductionTC}
Let $N,N'$ be tree-child networks on $X$ such that $N'$ is displayed by $N$. If a TCS $S$ reduces $N$, then $S$ also reduces $N'$.
\end{lemma}

Unlike when the input consists of only trees, a solution to {\sc Tree-child Network Hybridization} does not always exist when the input may also contain networks.

\begin{definition}
A set of tree-child networks~$\mathcal{N}$ are \emph{tree-child compatible} if there exists a tree-child network that displays all networks in~$\mathcal{N}$.
\end{definition}

Our next step, is to prove that there is a strong connection between tree-child compatibility and TCSs.

\begin{lemma}
 Let $\mathcal{N}$ be a set of tree-child networks on $X$. Then $\mathcal{N}$ is tree-child compatible iff there exists a TCS that reduces $\mathcal{N}$. Furthermore, if a solution exists, then $h_{\mathrm{tc}}(\mathcal{N})=s_{\mathrm{tc}}(\mathcal{N})$.
\end{lemma}
\begin{proof}
 Suppose that $\mathcal{N}$ is tree-child compatible. 
 Then there exists a tree-child network~$N$ that displays $\mathcal{N}$, with minimal reticulation number. 
 Now let $S$ be a TCS for $N$. By Lemma~\ref{lem:ContainmentImpliesReductionTC}, $S$ also reduces all displayed networks of $N$, and hence it reduces $\mathcal{N}$. Furthermore, the weight of $S$ is equal to the reticulation number of $N$ by Lemma~3 from \cite{YukiRemie}, 
 (originally proved slightly less strong in \cite{linz2019attaching}).
 
 Now suppose there exists a TCS $S$ that reduces $\mathcal{N}$. Let $N$ be the tree-child network that can be constructed from $S$. Then, by Lemma~\ref{lem:reductionImpliesContainment}, $N$ displays $\mathcal{N}$. Because $N$ is the network corresponding to $S$, the reticulation number of $N$ is equal to the weight of $S$.
\end{proof}

\subsection{The existence of a tree-child solution}
In the previous subsection, we have found a strong connection between {\sc Tree-child Network Hybridization} and {\sc Tree-child Weight} for feasible solutions. Not all inputs, however, are feasible. Here, we investigate the feasibility of inputs, and how to deal with infeasible inputs.

\begin{lemma}\label{lem:reducingReticCherries}
Let $N$ be a tree-child network with reticulated cherry $(x,y)$, then any TCS that reduces $N$ must contain the pair $(x,y)$.
\end{lemma}
\begin{proof}
Suppose $S$ is a TCS that reduces $N$. The only ways to reduce the reticulated cherry $(x,y)$ are by either reducing it directly with the pair $(x,y)$, or by first turning it into a cherry $\{x,y\}$ and then reducing it with a pair~$(x,y)$ or~$(y,x)$. This second option, however, leads to a contradiction: To make the reticulated cherry into a cherry, we must reduce a pair of the form $(x,\cdot)$; however, any sequence that includes $(x,\cdot)$ and later $(y,x)$ cannot be tree-child. 
\end{proof}

Using the connection between tree-child compatibility and the existence of TCSs, we can prove an obstruction to tree-child compatibility of Lemma~\ref{lem:reticCherryTCCompatible}. This obstruction will turn out to be quite valuable in the proofs in the rest of this paper, as it allows us to quickly check whether a set of networks is tree-child compatible.

 \begin{figure}[h]
    \centering
    \includegraphics[width=.6\textwidth]{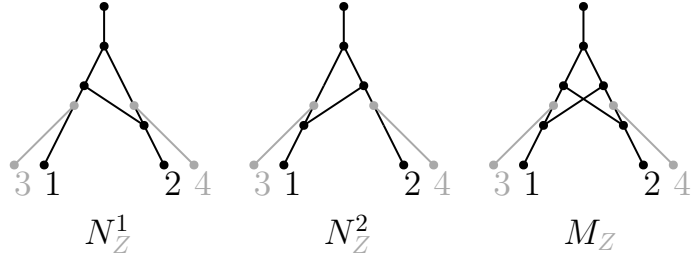}
    \caption{Two networks $\mathcal{N}=\{N^1,N^2\}$ that are tree-child incompatible (black parts only). The network $M$ displays $\mathcal{N}$, but it is not tree-child. By adding leaves $Z=\{3,4\}$ to $M$, we get the network $M_Z$ which is tree-child. Then, adding these leaves in the right places to $N^1$ and $N^2$, we get the set of networks $\mathcal{N}_Z\in\mathcal{N}^+$ on $X\cup Z$, that are displayed by the tree-child network $M_Z$.} 
    \label{fig:TCCompatibility}
\end{figure}

\begin{lemma}\label{lem:reticCherryTCCompatible}
Let~$N_1, N_2$ be tree-child networks on the same set of leaves~$X$.
For any pair of leaves~$x,y$, if~$N_1$ contains the reticulated cherry~$(x,y)$ and~$N_2$ contains the reticulated cherry~$(y,x)$, then~$N_1$ and~$N_2$ are not tree-child compatible.
\end{lemma}
\begin{proof}
Let $N$ be a tree-child network that displays both $N_1$ and $N_2$. Then any TCS $S$ for $N$ reduces both $N_1$ and $N_2$. By Lemma~\ref{lem:reducingReticCherries}, the sequence $S$ must contain the pair $(x,y)$, because $N_1$ has the reticulated cherry $(x,y)$; similarly, $S$ must contain $(y,x)$. This means $S$ is a TCS, but it includes both pairs $(x,y)$ and $(y,x)$, a contradiction. Hence we conclude that $N_1$ and $N_2$ are not tree-child compatible.
\end{proof}

Even if an input is infeasible, we still desire a network that displays all input networks. For this purpose, we can relax the tree-child constraint on output (and input) of the {\sc Tree-child Network Hybridization} problem, giving rise to the following problem.
\medskip\\
\fbox{
\parbox{0.8\linewidth}{
{\sc Network Hybridization}\\
{\bf Input:} A set of rooted networks $\mathcal{N}$ on $X$.\\
{\bf Output:} A network that displays $\mathcal{N}$ with minimal reticulation number.}
}
\medskip\\
This problem can be viewed as the natural extension of the classic {\sc Hybridization} problem for trees. Linz and Semple show that {\sc Hybridization} can be solved by adding leaves in the right place to all input trees, and then solving {\sc Tree-child Hybridization} \cite{linz2019attaching}. This also holds for the network versions of these problems, as the solution to {\sc Network Hybridization} can be made tree-child by adding leaves, and all networks displayed by a tree-child network are tree-child networks as well (Figure~\ref{fig:TCCompatibility}).

\section{An Algorithm for {\sc Tree-child Network Hybridization}}\label{sec:Algorithm}
In this section, we give an FPT algorithm for {\sc Tree-child Network Combination}. 
We extend the algorithm given in \cite{van2019practical} by allowing for inputs to be networks, and by looking for reducible pairs within networks rather than cherries in trees.
Given an input~$\cal{N}$ of tree-child networks, we first look for \emph{trivial} reducible pairs.
We show that it is safe to reduce trivial reducible pairs as soon as we encounter one, in any order.
We then branch on all possible non-trivial reducible pairs of the network, and by showing that the total number of possible reducible pairs at each branching point is at most~$8k$ for the reticulation number~$k$ of the optimal solution, we show that the running time of the algorithm is~$O((8k)^k\cdot \mathrm{poly}(|X|,|\mathcal{N}|)$.

\subsection{Counting Cherries}\label{sec:CountingCherries}
\subsubsection{Trivial pairs}
The algorithm in \cite{van2019practical} reduces \emph{trivial cherries} (a pair of leaves $\{x,y\}$ that appear as a cherry in any input tree containing $x$ and $y$) whenever possible. Here, only looking at trivial cherries is not sufficient. For an input of networks, we will need to reduce \emph{trivial reducible pairs} (trivial pairs for short) whenever possible. A trivial pair is a pair of leaves $(x,y)$ such that all networks either only have the leaf $y$, or they have a reducible pair~$(x,y)$. 
In the following two lemmas, we prove that it is safe to reduce such a pair as soon as we encounter one.

\begin{lemma}[Move to the left]\label{lem:Move1Left}
Let $\mathcal{N}=\{N_1,\ldots,N_I\}$ be a set of tree-child networks on a common set of leaves, and let $S(a,b)(x,y)S'$ be a TCS for $\mathcal{N}$. Suppose that for each $N\in\mathcal{N}S$ we have either $x$ is not a leaf in $N$, or $(x,y)$ is a reducible pair of $N$, and there is at least one network such that the latter holds. Then there is a TCS for $\mathcal{N}$ starting with $S(x,y)$ of length equal to that of $S(a,b)(x,y)S'$.
\end{lemma}
\begin{proof}
Suppose $b=x$. Then there must be a network in $\mathcal{N}S$ that has both the reducible pairs $(x,y)$ and $(a,x)$. This can only occur if $a=y$: as $x$ is the first as well as the second element of a reducible pair, it must form a cherry with another leaf, namely the leaf $y$. However, $S(y,x)(x,y)S'$ is not a TCS, which contradicts our assumption that $S(a,b)(x,y)S'$ is a TCS for $\mathcal{N}$.

Hence, for the rest of the proof, we assume $b\neq x$. In this case, $S(x,y)(a,b)S'$ is a TCS.
It remains to prove that it reduces $\mathcal{N}$. This is clear if $\{x,y\}\cap\{a,b\}=\emptyset$. 
Observe that~$a\neq y$, as otherwise~$S(a,b)(x,y)S'$ would not have been a TCS to begin with.
Therefore, we still need to check the cases $a=x$ and $b=y$. 

If $a=x$ and a network has both reducible pairs $(x,b)$ and $(x,y)$, then this network has a reticulation with reticulated cherries $(x,b)$ and $(x,y)$. The order of reducing these pairs obviously does not matter for such networks: both options remove the reticulation edges between the parents of $b$ and $y$, and the parent of $x$. For a network $N$ that only has the reducible pair $(x,y)$ after $S$ (and not $(x,b)$), the network $NS(x,y)(x,b)$ is a subnetwork of $NS(x,b)(x,y)=NS(x,y)$. This means $S(x,y)(x,b)S'$ also reduces $N$ \cite{YukiRemie}. Hence if $a=x$, the sequence $S(x,y)(a,b)S'$ is a TCS for $\mathcal{N}$.

Now suppose $b=y$. Let $N$ be a network that has both reducible pairs $(a,y)$ and $(x,y)$. 
But all tree nodes of~$N$ are of outdegree-$2$; this implies that every leaf can be the second coordinate of at most one reducible pair.
Therefore such a network cannot exist, and thus this case is not possible.
\end{proof}

\begin{lemma}[Trivial pair reduction]\label{lem:trivialPairs}
Let $\mathcal{N}=\{N_1,\ldots,N_I\}$ be a set of tree-child networks on a common set of leaves such that there exists a TCS $SS'$ for $\mathcal{N}$. 
Suppose $x,y$ are leaves such that for each $N\in\mathcal{N}S$ we have either $x$ not in $N$, or $(x,y)$ a reducible pair of $N$,
and there is at least one network such that the latter holds. 
Then there exists a TCS $S(x,y)S''$ of length equal to $SS'$ that reduces $\mathcal{N}$, or if $y$ is forbidden after $S$ and there is a sequence of the form $S(y,x)S'''$ of the same length as $SS'$ that reduces $\mathcal{N}$.
\end{lemma}
\begin{proof}
To reduce a network with reducible pair $(x,y)$, the sequence $S'$ must contain either $(x,y)$ or $(y,x)$. Let $S'_i$ be the first occurrence of such a pair. 

First suppose $S'_{i}=(x,y)$. Then for each intermediate set of networks $\mathcal{N}SS'_{[:j]}$ for $j<i$ we have that all the networks in the set either do not contain $x$, or have the reducible pair $(x,y)$. Hence, by repeated application of Lemma~\ref{lem:Move1Left}, there is a sequence $S(x,y)S''$ for $\mathcal{N}$. This sequence has the same length as $SS'$, because it is simply a reordering of the pairs.

Now suppose $S'_i=(y,x)$, then $x$ cannot have been the first coordinate in any pair of $S$, so all networks in $\mathcal{N}S$ contain $x$. Furthermore, $S'$ does not contain the pair $(x,y)$, as this would violate the assumption that $SS'$ is a TCS. Hence, each network in $\mathcal{N}S$ has a cherry or reticulated cherry on $x$ and $y$, which is ultimately reduced by a pair $(y,x)$ in $S'$. Suppose a network $N\in\mathcal{N}S$ does not have the cherry $\{x,y\}$. Then it has the reticulated cherry $(x,y)$. To make this into a cherry, so that it can be reduced by $(y,x)$, the sequence must first contain a pair of the form $(x,z)$. However, this implies $S'$ first contains $(x,z)$ and then $(y,x)$, which contradicts the fact that $SS'$ is a TCS. Hence, we may assume that all networks in $\mathcal{N}S$ have the cherry $\{x,y\}$.

If $y$ is not forbidden after $S$, we can switch the roles of $x$ and $y$ in the remaining part of the sequence $S'$ to get a new TCS $SS^*$ for $\mathcal{N}$. In $S^*$, the first occurrence of $(x,y)$ or $(y,x)$ is $S^*_{i}=(x,y)$, and we are in the previous case. If $y$ is forbidden after $S$, repeated application of Lemma~\ref{lem:Move1Left} on $SS'$ and $S'_i$ gives a sequence $S(y,x)S'''$ for $\mathcal{N}$.
\end{proof}

\subsubsection{Bounding reducible pairs in networks with all leaves}
In the algorithm in \cite{van2019practical}, a bound on the number of cherries after having reduced all trivial cherries was required to compute the running time. 
Here, we require something similar; we require a bound on the number of reducible pairs after we have reduced all the trivial pairs. \cite{van2019practical} prove such bounds by first focusing on the case where all input trees have the same leaf set. We do the same, by first focusing on the case where all input networks have the same leaf set. 

Let~$\cal{N}$ be a set of networks.
Then the \emph{set of displayed trees of}~$\cal{N}$ is the set of all trees that are displayed by the networks of~$\cal{N}$.

\begin{lemma}\label{lem:TrivPair}
Let $\mathcal{N}=\{N_1,\ldots,N_I\}$ be a set of tree-child networks on a common set of leaves such that there exists a TCS $S$ for $\mathcal{N}$. If $\mathcal{N}$ does not contain any trivial pairs, then the set of displayed trees of $\mathcal{N}$ has no trivial cherries.
\end{lemma}

\begin{lemma}[\cite{van2019practical} Lemma~10]
  \label{lem:4k-cherries}
  Let $\mathcal{T}$ be a set of phylogenetic trees with leaf set $X$ such that there is a tree-child network $N$ with $k$ reticulations that displays $\mathcal{T}$. If $\mathcal{T}$ has no trivial cherries, then the total number of cherries of the trees in $\mathcal{T}$ is at most $4k$.
\end{lemma}

Lemmas~\ref{lem:TrivPair} and~\ref{lem:4k-cherries} gives the bound on the number of reducible pairs for networks with common leaf sets.

\begin{lemma}
  \label{lem:8k-reduciblePairs}
  Let $\mathcal{N}$ be a set of tree-child networks with leaf set $X$ such that there is a tree-child network $N$ with $k$ reticulations that displays $\mathcal{N}$. If $\mathcal{N}$ has no trivial pairs, then the total number of reducible pairs of the networks in $\mathcal{N}$ is at most $8k$.
\end{lemma}
\begin{proof}
Each reducible pair of a network is a cherry in one of its displayed trees, and the set of displayed trees is displayed by the solution network $N$ as well. Hence, by Lemma~\ref{lem:4k-cherries}, there are at most $8k$ reducible pairs in the trees, and therefore at most $8k$ reducible pairs in the networks.
\end{proof}

\subsubsection{Bounding reducible pairs in general}

Recall that the algorithm will build a TCS by successively appending reducible pairs; it terminates upon finding the shortest possible sequence that reduces all the input networks.
In the process, it branches on all possible non-trivial pairs that the input network may have.
Depending on the sequence that is being built, it is possible that leaves that exist in some of the input networks (after reduction by the existing sequence) may have already been deleted from others.
Here, we show that even for these instances, it is still the case that the number of possible reducible pairs that we can branch on is bounded by~$8k$.
This result follows directly from Lemma~7 of \cite{van2019practical}: we change the wording of the statement slightly to accommodate for network inputs.

\begin{lemma}\label{lem:maintain-8k-pairs}
 Let $\mathcal{N}$ be a set of tree-child networks on $X$, and let $S = (x_1, y_1), (x_2, y_2), \ldots, (x_r, y_r) $be a TCS for $\mathcal{N}$ with weight $k$.
 For any $j \in [r] \cup \{0\}$, either there exists a trivial pair of $\mathcal{N}S_{[:j]}$, or $\mathcal{N}S_{[:j]}$ has at most $8k$ reducible pairs.
\end{lemma}
 The idea of the proof is as follows. Let $j$ be such that $\mathcal{N}S_{[:j]}$ has no trivial pairs. Then we find a set of tree-child networks $\hat{\mathcal{N}}_j$ on $X$ with the same set of reducible pairs as $\mathcal{N}S_{[:j]}$ and tree-child hybridization number at most $k$. By Lemma~\ref{lem:8k-reduciblePairs}, this shows that $\mathcal{N}S_{[:j]}$ has at most $8k$ reducible pairs. 
 
 The set of networks is constructed by adding back each missing leaf to each network in $\mathcal{N}S_{[:j]}$ at the root. The order in which they are placed at the root is the same as the order in which these leaves appear as first element in $S_{[:j]}$. 
 Now, we may construct a TCS of the same weight as~$S$ that reduces this set of networks. 
 By first reducing the part that corresponds to the part in $\mathcal{N}S_{[:j]}$, and then the leaves placed by the root, we have a TCS that reduces~$\hat{\mathcal{N}}_j$ of weight at most $k$:
 $$(x_{j+1}, y_{j+1}), (x_{j+2}, y_{j+2}), \ldots, (x_r, y_r),  (x_1, y_r), (x_2, y_r), \dots, (x_j, y_r).$$
 An example of the corresponding networks and their embeddings can be found in Figure~\ref{fig:Maintain-8k-pairs}.
 
 \begin{figure}[h]
    \centering
    \includegraphics[width=\textwidth]{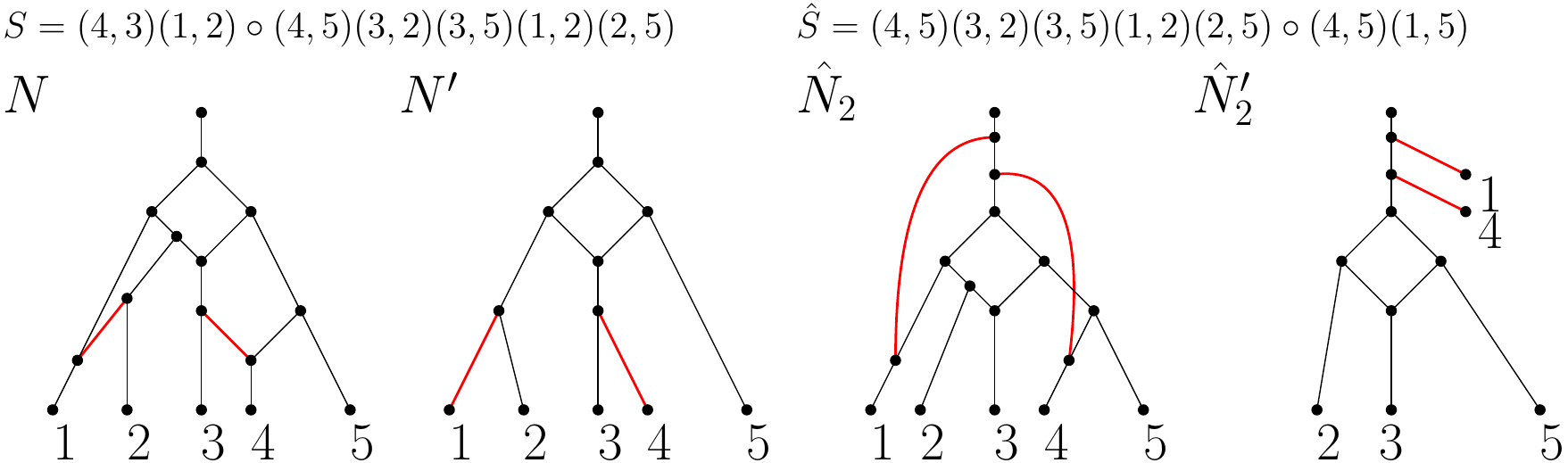}
    \caption{A set~$\mathcal{N} = \{N,N'\}$ of tree-child networks on the set of taxa~$\{1,2,3,4,5\}$, with a TCS~$S$ that reduces it.
    A set~$\hat{\mathcal{N}_2} = \{\hat{N}_2, \hat{N}'_2\}$ of tree-child networks obtained by reducing the first two elements of~$S$ from~$\mathcal{N}$, and reattaching the tail of the deleted edges (red edge) to the root edge, in the order that they were deleted in (as explained in the sketch proof of Lemma~\ref{lem:maintain-8k-pairs}).
    The sequence~$\hat{S}$ is a TCS of the same weight as~$S$, obtained from~$S$ by deleting the first two elements and appending these two elements to the end of the sequence, for which we replace the second coordinate of the elements by~$5$ (the leaf that appears as the second coordinate element in the last element of~$S$.
    }
    \label{fig:Maintain-8k-pairs}
\end{figure}

\begin{procedure}[h]
  \caption{TreeChildSequence($\mathcal{N},S,k$)\label{alg:tree-child-sequence}}
  \KwIn{A collection $\mathcal{N}$ of tree-child networks, a partial TCS $S$, an integer $k$\;}
  \KwOut{An optimal TCS $S S'$ of weight at most $k$ for $\mathcal{N}$ if such
    a sequence exists; $\Fail$ otherwise\;} 
%
 
 \While{There exists a trivial pair $(x,y)$ in $\mathcal{N}S$ with $y$ not forbidden by $S$\label{lin:LoopTrivialCherryReductionStart}}
 { 
    Set $S = S(x,y)$\;
 } \label{lin:LoopTrivialCherryReductionEnd}
 Set $\mathcal{N}' = \mathcal{N}S$\;
 \If{some network in $\mathcal{N}'$ has a cherry $(x,y)$ with $x,y$ forbidden or a reticulated cherry $(x,y)$ with $y$ forbidden}
 {
   \Return {$\Fail$\;}   \label{lin:FailForbiddenCherry}
 } 
 \Else 
 {
   Set $n' = |\{x \in X: x$ is a leaf in $\mathcal{N}'\}|$\; \label{lin:DefineNumberRemainingLeaves}
   Set $k' = |S| - |X| + n'$\;
   \label{lin:DefineCurrentWeight}
   Set $C =  \{(x,y) \mid (x,y) \text{ is a reducible pair of some network in } \mathcal{N}'\}$\;\label{lin:DefineCherrySet}
   \uIf{$|C| == 0$}
   {
       \Return $S$\;\label{lin:SolutionFound}
   }
   \uElseIf{$|C| > 8k$ or $k' \geq k$}
   {
       \Return{$\Fail$\;}\label{lin:FailTooManyCherriesOrReticulations}
   }
   \Else 
   {
       Set $S_{opt} = \Fail$\;\label{lin:RecursiveSolutionStart}
        \ForEach{$(x,y) \in C$ with $y$ not forbidden by $S$}{    
          Set $S_{temp} = \TreeChildSequence(\mathcal{N}, S(x,y), k)$\;
          \If{$S_{temp} \neq \Fail$ and ($S_{opt} = \Fail$  or ($S_{opt} \neq \Fail$ and $w(S_{temp}) < w(S_{opt})$))}
          {
            Set $S_{opt} = S_{temp}$\;
          }
        }
        \Return{$S_{opt}$\;}\label{lin:RecursiveSolutionEnd}
     }
   }
 
\end{procedure}

\subsection{Adapting the algorithm}\label{sec:Algorithms}
Our algorithms are exactly the same as those presented in \cite{van2019practical}, except for the following changes.

\begin{itemize}
    \item The input set of trees $\mathcal{T}$ is changed into an input set of tree-child networks $\mathcal{N}$;
    \item trivial cherries are now trivial pairs;
    \item In line~4, the stop condition of a non-pickable reticulated cherry is added;
\end{itemize}

The first change is necessary for the algorithm to take an input consisting of networks. The second change is necessary as not all reducible pairs are cherries anymore, when the input consists of networks. The while-loop that reduces all the trivial pairs is still correct in the algorithm, because there is an optimal sequence that first reduces all trivial pairs (Lemma~\ref{lem:trivialPairs}). The last change makes sure we stop when the reduced input $\mathcal{N}S$ cannot be fully reduced using a TCS that can be appended after the prefix $S$. 

Otherwise, the algorithm is still correct. Indeed, the algorithm branches over all non-trivial pairs, to find a shortest sequence that reduces all input networks; and this shortest sequence corresponds to a network with minimal reticulation number that displays all input networks. Furthermore, the running time follows as each branch-out is over at most $8k$ pairs, and the search depth is at most $k$.

\begin{theorem}
Let $\mathcal{N}$ be a set of tree-child networks on a set of taxa $X$. If there exists a tree-child network with at most $k$ reticulations that displays $\mathcal{N}$, then it can be found in $O((8k)^k\cdot\mathrm{poly}(|X|,|\mathcal{N}|))$ time using $\textsc{TreeChildNetwork}(\mathcal{N},k)$.
\end{theorem}

\begin{procedure}
  \caption{TreeChildNetwork($\mathcal{N},k$)\label{alg:tree-child-network}}
  \KwIn{A collection $\mathcal{N}$ of tree-child networks, an integer $k$\;}
  \KwOut{ A tree-child phylogenetic network $N$ on $X$ that displays $\mathcal{N}$ with reticulation number at most $k$, if such a network exists; otherwise {\sc None}\;}
  
  Set $S =  \TreeChildSequence(\mathcal{N}, \emptyset, k)$\;
  \If{$S == \Fail$}{
    \Return {$\None$\;}
  }\Else{
    Set $N = \ConstructNetworkFromSequence(S)$\;
    \Return $N$\;
  } 
\end{procedure}

\section{Discussion}\label{sec:Discussion}
In this paper, we have introduced {\sc Network Hybridization}, the problem of finding a network with minimal reticulation number that displays a set of networks.
We showed that the {\sc Tree-child Network Hybridization} problem, in which we restrict our inputs and output to be tree-child networks, can be solved by making slight adjustments to the FPT algorithm presented in \cite{van2019practical}.

In practice, our algorithm can be sped up using the heuristic improvement that was introduced in \cite{van2019practical}.
We may consider branch reduction, in which we ignore parts of the search tree where no better solution can be found.

For this problem, FPT is essentially the best we can do, because solving the {\sc Network Hybridization} problem for an input set of tree-child networks is NP-hard. This follows from the fact that it is already NP-hard for an input set of trees. 
It has recently been shown that if all level-$(k-1)$ subnetworks of a level-$k$ tree-child networks are given, this network can be constructed in polynomial time \cite{murakami2019reconstructing}. In other words, the {\sc Tree-child Network Hybridization} problem is easy to solve when we are given all level-$(k-1)$ subnetworks of a level-$k$ network. 
This suggests that the problem becomes easy if the difference in reticulation number between the inputs and the output network is bounded. 
We wonder if this is still true for networks that are not tree-child, and therefore it would be interesting to see whether the {\sc Hybridization} problem is FPT with this difference in reticulation number as parameter. And, if this is the case, whether the current algorithm can be proven to have this running time. 

Recall that a TCS is a sequence of ordered pairs with two conditions imposed on them: the first condition ensures that we obtain a network from the sequence upon using the {\sc ConstructNetworkFromSequence} algorithm; the second condition ensures that the network we obtain is tree-child.
Upon removing this second condition from sequences of ordered pairs, we obtain what is called a \emph{cherry-picking sequence} \cite{YukiRemie}.
We call networks that can be reduced by a cherry-picking sequence an \emph{orchard} network.
A natural extension of the results we have presented in this paper would be to consider the following problem.

\medskip
\fbox{
\parbox{0.8\linewidth}{
{\sc Orchard Network Hybridization}\\
{\bf Input:} A set of orchard networks $\mathcal{N}$ on $X$.\\
{\bf Output:} An orchard network that displays $\mathcal{N}$ with minimal reticulation number.}
}
\medskip

Ideally, we would simply relieve Algorithm {\sc TreeChildSequence} of the tree-child condition to obtain an algorithm which would also work for orchard networks.
However simply doing so could potentially result in a much higher running time, as we do not have a bound on the number of reducible pairs for orchard networks (see Figure~\ref{fig:CherryCountCPS}).
Nevertheless, an important consequence of solving {\sc Orchard Network Hybridization} is that since tree-child networks are orchard networks, we would be able to obtain a better upper bound for the hybridization number (i.e. the reticulation number of an optimal solution for {\sc Network Hybridization}).

\begin{figure}[h]
    \centering
    \includegraphics[width=.6\textwidth]{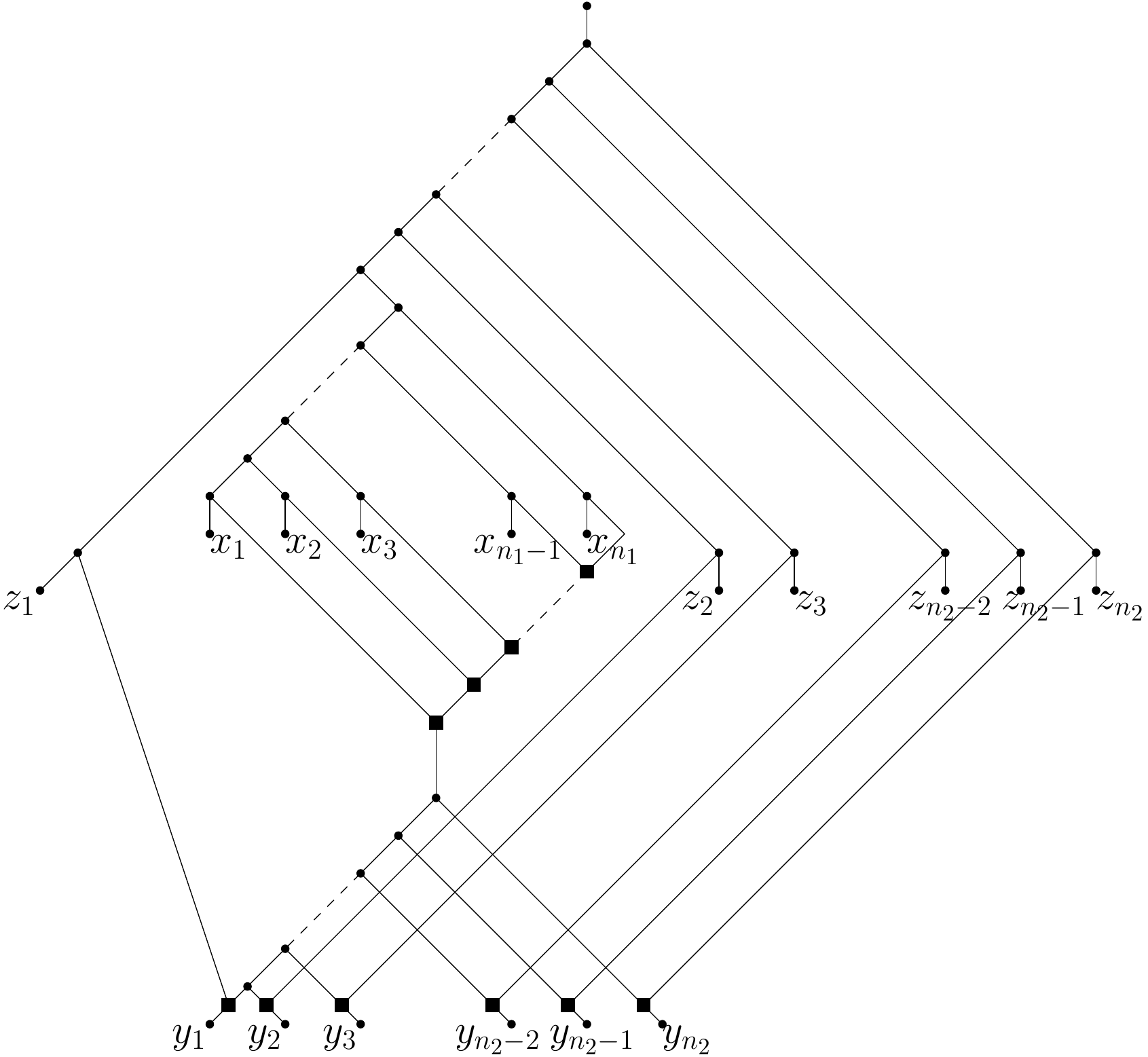}
    \caption{An orchard network with $n_1+n_2-1$ reticulations such that the set of displayed trees have at least $n_1n_2$ cherries.}
    \label{fig:CherryCountCPS}
\end{figure}

\bibliographystyle{plain}
\bibliography{bibliography}
\end{document}